\documentclass[twocolumn,reprint,pra,aps,superscriptaddress,nofootinbib,9pt]{revtex4-1}

\usepackage[utf8]{inputenc}
\usepackage[english]{babel}
\usepackage{braket}
\usepackage{amscd}
\usepackage{amsmath}
\usepackage{hyperref}
\usepackage{amscd}
\usepackage{lipsum}
\usepackage{amsthm}
\usepackage{amssymb}
\usepackage{latexsym}
\usepackage{epic}

\newtheorem{theorem}{Theorem}
\newtheorem{corollary}{Corollary}[theorem]
\newtheorem{lemma}{Lemma}
\newtheorem{definition}{Definition}
\newtheorem{proposition}{Proposition}
\newtheorem{cex}{Counter-Example}

\begin{document}


\title{Necessary and sufficient condition for the existence of Schmidt decomposition in multipartite Hilbert spaces}


\author{Spandan Das}
\email{spandandas94@gmail.com}
\affiliation{Department of Mathematics and Statistics, Indian Institute of Science Education and Research (IISER), Kolkata, India}
\author{Goutam Paul}
\email{goutam.paul@isical.ac.in}
\affiliation{Cryptology and Security Research Unit (CSRU), R. C. Bose Centre for Cryptology and  Security, Indian Statistical Institute, Kolkata, India}



\begin{abstract}
Pati (Physics Letters A, 2000) derived a sufficient condition for the existence of Schmidt decomposition in tripartite Hilbert spaces. In this paper, we show that the condition is erroneous by demonstrating some counter-examples. Moreover, we suitably modify the condition and provide a correctness proof. We also show for the first time how this can be generalized to $n$-partite Hilbert spaces for any $n\geq 2$. Finally, we prove that this condition is also a necessary condition.
\end{abstract}


\maketitle


\section{Introduction}
Schmidt decomposition was originally defined for pure states in a bipartite Hilbert space; since for any pure state $\ket{x}$ in a two-part system $A\otimes B$, we can always find orthonormal bases $\{\ket{i_A}\}_{i=1}^{dim(A)}\subseteq A$ and $\{\ket{i_B}\}_{i=1}^{dim(B)}\subseteq B$ such that $\ket{x}$ can be written as,
\begin{equation*}
\ket{x} = \sum_{i=1}^n \lambda_i \ket{i_A}\ket{i_B},
\end{equation*}
where $n=min\{dim(A),dim(B)\}$ and $\lambda_i \in \mathbb{C}$ \cite{2,3,4}. Schmidt decomposition is extremely helpful for bipartite systems since it reduces the number of terms needed to represent a pure state through orthogonal bases from $dim(A)\times dim(B)$ to just $min\{dim(A),dim(B)\}$. Thus all the information contained in the pure state can be confined in a much smaller space. This is why Schmidt decomposition has enormous application in the study of Quantum Information. It would be tremendously helpful if Schmidt decomposition existed for general $n$-partite system $\otimes_{i=1}^n A_i$, simply because, we would have to deal with just $min\{dim(A_i)|i=1,\dots,n\}$ terms instead of $\prod_{i=1}^n dim(A_i)$ terms then. But sadly this is not true. Schmidt decomposition woefully fails for just tripartite systems \cite{2,5}. But certainly, some pure states in a general $n$-partite Hilbert space are Schmidt decomposable. 

Pati derived a sufficient condition for the existence of Schmidt decomposition in tripartite Hilbert spaces~\cite{1}, but the condition was erroneous, as we have proved in this article. We also showed that with a little modification, the result holds not only for tripartite systems, but for any $n$-partite system in general; and it is actually a necessary and sufficient condition.\par

A different direction of research is to find the equivalence class of a pure state in a composite quantum system \cite{6,7,8,9,10,11,12,13,14,15,16,17,18}, where two pure states are called equivalent if they are related by a unitary transformation which factorises into separate unitary transformations on the component parts (i.e., a \textit{local} unitary transformation). A remarkable work in this direction is due to Carteret, Higuchi and Sudbery~\cite{13}, where a nice way of describing any pure state in a composite system through orthonormal bases of the component systems is derived. An interesting property of this canonical form is that the coefficient of each term is either $0$ or real. This property is also seen in the Schmidt decomposition of a pure state in a bipartite system. Hence the authors have named the canonical form as ``generalised Schmidt decomposition''. But this representation is clearly different from what we would understand as the Schmidt decomposition of a pure state. 

It is to be noted that, by our definition, not every pure state in an $n$-partite system ($n>2$) has Schmidt decomposition, but, by the definition of the paper \cite{13}, every pure state in a general $n$-partite system has it. Thus to avoid confusion, we have provided definitions of certain terms in order to clarify what we have meant by them.\par

For the rest of this article, $\delta_{ij}$ is defined as,
\begin{eqnarray}\label{delta}
&\delta_{ij}=
\begin{cases}
1;\text{ if }i=j,\\
0;\text{ otherwise.}
\end{cases}
\end{eqnarray}

\section{Useful Definitions}
\begin{definition}\label{nph}\textbf{n-partite Hilbert space}\\
	A system $A$ is called an \textbf{n-partite Hilbert space} if it is a tensor product of $n$ finite dimensional Hilbert spaces $A_1, A_2,\dots,A_n$ and inner product on $A$, i.e., $Ip : A \times A \longrightarrow \mathbb{C}$ is defined as,
	\begin{align*}
	& Ip(\sum_{i_1,\dots,i_n}a_{i_1,\dots,i_n}\otimes_{k=1}^{n}\ket{v_{i_k}^{A_k}} , \sum_{j_1,\dots,j_n}b_{j_1,\dots,j_n}\otimes_{k=1}^{n}\ket{w_{j_k}^{A_k}})\nonumber\\
	=& \sum_{i_1,\dots,i_n}\sum_{j_1,\dots,j_n} a_{i_1,\dots,i_n}^*b_{j_1,\dots,j_n}\prod_{k=1}^{n}\braket{v_{i_k}^{A_k}|w_{j_k}^{A_k}}\nonumber,
	\end{align*}
	where $\ket{v_{i_k}^{A_k}}, \ket{w_{j_k}^{A_k}} \in A_k \quad \forall i_k,j_k$; $k=1,\dots,n$.
\end{definition}

\begin{definition}\label{pip}\textbf{Partial inner product}\\
	Let $A_1,A_2,\dots,A_n$ be finite dimensional Hilbert spaces. Also let $\ket{x} =\sum_{i_1,\dots,i_n}a_{i_1,\dots,i_n}\ket{w_{i_1}^{A_1}}\dots\ket{w_{i_n}^{A_n}} \in \otimes_{j=1}^{n}A_j$ $(\ket{w_{i_j}^{A_j}} \in A_j \quad \forall j)$ and $\ket{v^{A_k}} \in A_k$ be pure states. Then the \textbf{partial inner product} $(_{A_k}\braket{|})$ of  $\ket{v^{A_k}}$ and $\ket{x}$ is an operator,
	$$
	_{A_k}\braket{|} : A_k \times \otimes_{j=1}^{n}A_j, \longrightarrow \mathop{\otimes_{j=1}^{n}}_{j\neq k}A_j,
	$$
	such that
	\begin{align*}
	&_{A_k}\braket{v^{A_k}|\sum_{i_1,\dots,i_n}a_{i_1,\dots,i_n}\ket{w_{i_1}^{A_1}}\dots\ket{w_{i_n}^{A_n}}}\\ =&\sum_{i_1,\dots,i_n}a_{i_1,\dots,i_n}\braket{v^{A_k}|w_{i_k}^{A_k}}\mathop{\prod_{j=1}^n}_{j\neq k}\ket{w_{i_j}^{A_j}}.
	\end{align*}
\end{definition}

\begin{definition}\label{pss}\textbf{Partially separable state}\\
	Let $A_1,A_2,\dots,A_n$ be finite dimensional Hilbert spaces. A pure state $\ket{x} \in \otimes_{j=1}^{n}A_j$ is called a \textbf{partially separable state} if it can be written in the form
	\begin{align*}
	&\ket{x} = \ket{\psi}_{A_1,A_2,\dots,A_{k-1}} \otimes \ket{\phi}_{A_{k},A_{k+1},\dots,A{n}}\\
	&\text{for some}\quad k \in \{2,\dots,n\},
	\end{align*}
	where $\ket{\psi}_{A_1,A_2,\dots,A_{k-1}} \in \otimes_{j=1}^{k-1}A_j$ and $\ket{\phi}_{A_{k},A_{k+1},\dots,A{n}} \in \otimes_{j=k}^{n}A_j$ are pure states.
\end{definition}

\begin{definition}\label{css}\textbf{Completely separable state}\\
	Let $A_1,A_2,\dots,A_n$ be finite dimensional Hilbert spaces. A pure state $\ket{x} \in \otimes_{j=1}^{n}A_j$ is called a \textbf{completely separable state} if it can be written in the product form, i.e.,
	\begin{equation*}
	\ket{x} = \otimes_{j=1}^{n}\ket{\psi^{A_j}},
	\end{equation*}
	where $\ket{\psi^{A_j}} \in A_j$ for each $j=1,\dots,n$.
\end{definition}

\section{Analysis of an existing result}
A paper by Pati~\cite{1} described a single theorem for existence of Schmidt decomposition of a pure state in a tripartite Hilbert space. The Theorem of Pati~\cite{1} is stated below.

\begin{proposition}\label{wrth}
	For any state $\ket{\psi}_{ABC} \in \mathbb{H}_A \otimes \mathbb{H}_B \otimes \mathbb{H}_C$ of a tripartite system (Def[\ref{nph}]), let dim $\mathbb{H}_A = N_A$ is the smallest of $N_A,N_B,N_C$. If the ``partial inner product" (Def[\ref{pip}]) of the basis $\ket{u_i}_A$ with the state $\ket{\psi}_{ABC}$, i.e., $_A\braket{u_i|\psi}_{ABC} = \ket{\psi_i}_{BC}$ has Schmidt number (Def[\ref{sn}]) one, then the Schmidt decomposition for a tripartite system exists.
\end{proposition}
\subsection{Summary of the proof of~\cite{1}}
	The Proposition[\ref{wrth}] is proved by noting that the state $\ket{\psi}_{ABC}$ can be written as,
	\begin{equation*}
	\ket{\psi}_{ABC} = \sum_i \ket{u_i}_A\ket{\psi_i}_{BC},
	\end{equation*}
	where $\ket{\psi_i}_{BC}$ is a pure state in $B\otimes C$. It is claimed that $\{\ket{\psi_i}_{BC}\}_i$ is an orthogonal set in $B\otimes C$, but not necessarily normalized. Further by assumption of the theorem, it is written,
	\begin{equation*}
	\ket{\psi_i}_{BC}=\ket{\beta_i}_B\ket{\gamma_i}_C,\quad \text{for each }i.
	\end{equation*}
	Then the reduced density matrices $\rho_A,\rho_B,\rho_C$ of each of the subsystems $A,B,C$ respectively, is calculated by taking partial traces on the remaining two subsystems and using the trace equalities $tr_C(\ket{\gamma_i}_{CC}\bra{\gamma_j}) = q_i\delta_{ij}$, where $q_i = ||\gamma_i||^2$; and $tr_B(\ket{\beta_i}_{BB}\bra{\beta_j}) = r_i\delta_{ij}$, where $r_i = ||\beta_i||^2$. Here $\delta_{ij}$ is taken as in equation(\ref{delta}). Thus,
	\begin{align*}
	\rho_A &= \sum_i q_ir_i\ket{u_i}_{AA}\bra{u_i},\\
	\rho_B &= \sum_i q_ir_i\ket{\beta^{\prime}_i}_{BB}\bra{\beta^{\prime}_i},\\
	\rho_C &= \sum_i q_ir_i\ket{\gamma^{\prime}_i}_{CC}\bra{\gamma^{\prime}_i},
	\end{align*}
	where $\ket{\beta^{\prime}_i}_B=\frac{\ket{\beta_i}_B}{\sqrt{r_i}}$ and $\ket{\gamma^{\prime}_i}_C=\frac{\ket{\gamma_i}_C}{\sqrt{q_i}}$. Thus $\{\ket{\beta^{\prime}_i}_B\}_i$ and $\{\ket{\gamma^{\prime}_i}_C\}_i$ are orthonormal sets in $B$ and $C$ respectively. Thus $\ket{\psi}_{ABC}$ has Schmidt decomposition since it can be written as,
	\begin{equation*}
	\ket{\psi}_{ABC} = \sum_i \sqrt{q_ir_i}\ket{u_i}_A\ket{\beta^{\prime}_i}_B\ket{\gamma^{\prime}_i}_C,
	\end{equation*}
	where $\{\ket{u_i}_A\}_i$, $\{\ket{\beta^{\prime}_i}_B\}_i$, $\{\ket{\gamma^{\prime}_i}_C\}_i$ are orthonormal sets in $A$, $B$, $C$ respectively.
	
\subsection{Errors in the proof of~\cite{1}}
The proof that has been provided in the Paper\cite{1} in support of the Proposition[\ref{wrth}], requires some conditions that are not correct. We have enlisted these arguments and provided counter-example for each of them.
\begin{description}
	\item[Error 1] Let $\ket{\psi}_{ABC} = \sum_{ijk}a_{ijk}\ket{u_i}_A\ket{v_j}_B\ket{w_k}_C$. The partial inner product of the basis element $\ket{u_i}$ and the state $\ket{\psi}_{ABC}$ is a vector $\ket{\psi_i}_{BC}$ in the Hilbert space $\mathbb{H}_B \otimes \mathbb{H}_C$, which is spanned by basis vectors $\{\ket{v_j}_B \otimes \ket{w_k}_C\}$. Then $\ket{\psi_i}_{BC}$ can be written as,
	\begin{equation*}
	\ket{\psi_i}_{BC} = \sum_{jk} a_{ijk}\ket{v_j}_B \otimes \ket{w_k}_C,
	\end{equation*}
	where $\{\ket{\psi_i}_{BC}\}_i$ is an orthogonal basis set but need not be normalised.

	The above claim is made by the author \cite{1} which need not be true provided the choice of basis $\{\ket{u_i}\}_i$ is arbitrary.

	\begin{cex}
	Let $\mathbb{H}_A = \mathbb{H}_B = \mathbb{H}_C = \mathbb{H}$ where $\mathbb{H}$ is a two dimensional Hilbert space. Let
	\begin{equation*}
	\ket{\psi}_{ABC} = (\ket{0}+\ket{1}) \otimes \ket{0} \otimes \ket{0},
	\end{equation*}
	where $\{\ket{0},\ket{1}\}$ is an orthonormal basis for $\mathbb{H}$. Now let us observe that,
	\begin{align*}
	& \ket{\psi_0}_{BC} = _A\braket{0|\psi}_{ABC} = \ket{0} \otimes \ket{0},\\
	& \ket{\psi_1}_{BC} = _A\braket{1|\psi}_{ABC} = \ket{0} \otimes \ket{0}.
	\end{align*}
	Thus $\{\ket{0},\ket{1}\}$ satisfies the condition mentioned in the Proposition[\ref{wrth}], since both $\ket{\psi_0}_{BC}$ and $\ket{\psi_1}_{BC}$ are product states and hence equivalently have Schmidt number 1 by Lemma[\ref{sh1iffpro}]. But clearly they are not orthogonal since,
	\begin{equation*}
	_{BC}\braket{\psi_0|\psi_1}_{BC} = \braket{0|0}\braket{0|0} = 1.
	\end{equation*}
	\end{cex}
	Thus the claim is not always true.
	
	\item[Error 2] Let us assume also, that $\{\ket{\psi_i}_{BC}\}_i$ is an orthogonal set. Now,
	\begin{equation*}
	\ket{\psi_i}_{BC} = \ket{\beta_i}_B \otimes \ket{\gamma_i}_C,
	\end{equation*}
	according to the assumption of the Proposition[\ref{wrth}]. Now the author has used the trace equality $tr_C(\ket{\gamma_i}_{CC}\bra{\gamma_j}) = _C\braket{\gamma_i|\gamma_j}_C = q_i\delta_{ij}$, where $q_i = ||\gamma_i||^2$. Here $\delta_{ij}$ is taken as in equation(\ref{delta}). This also is not necessarily true always. The following counter-example proves this.

	\begin{cex}
	Let $\mathbb{H}_A = \mathbb{H}_B = \mathbb{H}_C = \mathbb{H}$ where $\mathbb{H}$ is a two dimensional Hilbert space. Let
	\begin{equation*}
	\ket{\psi}_{ABC} = \ket{0} \otimes \ket{0} \otimes \ket{0} + \ket{1} \otimes \ket{1} \otimes \ket{0},
	\end{equation*}
	where $\{\ket{0},\ket{1}\}$ is an orthonormal basis for $\mathbb{H}$. Now let us observe that,
	\begin{align*}
	& \ket{\psi_0}_{BC} = _A\braket{0|\psi}_{ABC} = \ket{0} \otimes \ket{0} = \ket{\beta_0}_B\otimes\ket{\gamma_0}_C,\\
	& \ket{\psi_1}_{BC} = _A\braket{1|\psi}_{ABC} = \ket{1} \otimes \ket{0} = \ket{\beta_1}_B\otimes\ket{\gamma_1}_C.
	\end{align*}
	Thus $\{\ket{0},\ket{1}\}$ satisfies the condition mentioned in the Proposition[\ref{wrth}], since both $\ket{\psi_0}_{BC}$ and $\ket{\psi_1}_{BC}$ are product states and hence equivalently have Schmidt number one by lemma[\ref{sh1iffpro}]. Also $_{BC}\braket{\psi_0|\psi_1}_{BC} = \braket{0|1}\braket{0|0} = 0$. Hence $\{\ket{\psi_0}_{BC},\ket{\psi_1}_{BC}\}$ is an orthogonal set. But note that $\ket{\gamma_0}=\ket{\gamma_1}=\ket{0}$. Thus,
	\begin{equation*}
	_{C}\braket{\gamma_0|\gamma_1}_C = \braket{0|0} = 1 \neq ||\gamma_0||^2 \delta_{01} = 0.
	\end{equation*}
	\end{cex}
	This proves that the trace equality is not necessarily true always.
\end{description}
\subsection{Errors in the theorem statement of~\cite{1}}
The claim made in the Proposition[\ref{wrth}] is not necessarily true always. This can be shown by the counter-example given below,

\begin{cex}
Let $\mathbb{H}_A = \mathbb{H}_B = \mathbb{H}_C = \mathbb{H}$ where $\mathbb{H}$ is a two dimensional Hilbert space. Let
\begin{equation*}
\ket{\psi}_{ABC} = \ket{0} \otimes \ket{0} \otimes \ket{0} + \ket{1} \otimes \ket{1} \otimes \ket{0},
\end{equation*}
where $\{\ket{0},\ket{1}\}$ is an orthonormal basis for $\mathbb{H}$.\par
Now let us observe that $dim(\mathbb{H}_A)=min\{dim(\mathbb{H}_A),dim(\mathbb{H}_B),dim(\mathbb{H}_C)\}=2$ and $\exists$ orthonormal basis $\{\ket{0},\ket{1}\}\subseteq \mathbb{H}_A$ such that,
\begin{align*}
_{H_A}\braket{0|\psi}_{ABC}&=\ket{0}\otimes\ket{0},\\
_{H_A}\braket{1|\psi}_{ABC}&=\ket{1}\otimes\ket{0},
\end{align*}
i.e., the sufficient condition for the existence of Schmidt decomposition which was stated in Proposition[\ref{wrth}], is satisfied. Hence by Proposition[\ref{wrth}], $\ket{\psi}_{ABC}$ should have Schmidt decomposition.
\par
Now let us observe that $\ket{\psi}_{ABC}$ is not completely separable (Def[\ref{css}]). This is because when we observe $\ket{\psi}_{ABC}$ as a state of a bipartite Hilbert space $\mathbb{H}_A\otimes \mathbb{H}_D$, where $\mathbb{H}_D = \mathbb{H}_B\otimes \mathbb{H}_C$, the representation
\begin{equation*}
\ket{\psi}_{ABC} = \ket{0}_{\mathbb{H}_A}\otimes(\ket{0}\ket{0})_{\mathbb{H}_D} + \ket{1}_{\mathbb{H}_A}\otimes(\ket{1}\ket{0})_{\mathbb{H}_D}
\end{equation*}
is actually Schmidt decomposition of $\ket{\psi}_{ABC}$ in the Hilbert space $\mathbb{H}_A\otimes \mathbb{H}_D$; since $\ket{0}\ket{0}$ and $\ket{1}\ket{0}$ are orthonormal states in $\mathbb{H}_D$. Since its Schmidt decomposition has more than one nonzero terms, Schmidt number (Def[\ref{sn}]) of $\ket{\psi}_{ABC}$ is greater than 1. This implies $\ket{\psi}_{ABC}$ is not separable as a state of $\mathbb{H}_A\otimes \mathbb{H}_D$ by lemma[\ref{sh1iffpro}]. But then $\ket{\psi}_{ABC}$ is not completely separable when viewed as a state of $\mathbb{H}_A\otimes \mathbb{H}_B\otimes \mathbb{H}_C$. But we also see that $\ket{\psi}_{ABC}$ is partially separable (Def[\ref{pss}]) since it can be written as,
\begin{equation*}
\ket{\psi}_{ABC} = \big{(} \ket{0}\ket{0} + \ket{1}\ket{1} \big{)}\otimes \ket{0}.
\end{equation*}
Thus $\ket{\psi}_{ABC}$ is partially separable but not completely separable. By corollary[\ref{psc}], $\ket{\psi}_{ABC}$ does not have a Schmidt decomposition though Proposition[\ref{wrth}] claims that it does. Thus Proposition[\ref{wrth}] is disproved.
\end{cex}

\section{Correction of the previous result}
In this and the following section, we would like to provide the proof for a modified version of the previous Proposition[\ref{wrth}] which works equally for any $n$-partite Hilbert space. We will also show that the conditions mentioned are actually necessary and sufficient for a pure state to have Schmidt decomposition. But first let us clarify what we mean by Schmidt decomposition of a pure state in an n-partite Hilbert space.

\begin{definition}\label{shdb}\textbf{Schmidt decomposition for a bipartite Hilbert space}\\
	Any pure state $\ket{x}\in A \otimes B$, where $A$ and $B$ are finite dimensional Hilbert spaces, can be written as
	\begin{equation*}
	\ket{x} = \sum_{i=1}^{n} \lambda_{i}\ket{i_A}\ket{i_B},
	\end{equation*}
	where $n=min\{dim(A),dim(B)\}$, $\{i_A\}_{i=1}^n \subseteq A$ and $\{i_B\}_{i=1}^n \subseteq B$ are orthonormal sets. Also given $\ket{x}$, the set $\{\lambda_i \in [0,\infty);i=1,\dots,n\}$ as a multi-set is unique. This decomposition of a pure state is known as the \textbf{Schmidt decomposition}. 
\end{definition}

\begin{definition}\label{sn}\textbf{Schmidt number}\\
	Let $A,B$ be finite dimensional Hilbert spaces. The \textbf{Schmidt number} of a pure state $\ket{x} \in A\otimes B$ is the number of nonzero terms in its Schmidt decomposition.
\end{definition}

\begin{definition}\label{shdn}\textbf{Schmidt decomposition for an n-partite Hilbert space}\\
	A pure state $\ket{x} \in \otimes_{j=1}^{n}A_j$, where $A_j$ is a finite dimensional Hilbert space for each $j$ is said to have a Schmidt decomposition if  $\exists$ orthonormal sets $\{\ket{u_i^{A_j}}\}_{i=1}^m \subseteq A_j$ for each $j$ $(m=min\{dim(A_j)|j=1,\dots,n\})$ such that $\ket{x}$ can be written as
	\begin{equation*}
	\ket{x} = \sum_{i=1}^m \lambda_i \otimes_{j=1}^n \ket{u_i^{A_j}},
	\end{equation*}
	where $\lambda_i\in \mathbb{C}\quad \forall i$.
\end{definition}

Let $\ket{x} \in \otimes_{j=1}^{n}A_j$ where $A_j$ is a finite dimensional Hilbert space with dimension $N_{A_j}$ for each $j$.
\begin{lemma}\label{csl}
	If $\ket{x}$ is completely separable (Def[\ref{css}]), then $\ket{x}$ has a Schmidt decomposition.
\end{lemma}

\begin{proof}
	If $\ket{x}$ is completely separable then it can be written as,
	\begin{align*}
	&\ket{x}=\otimes_{j=1}^n \ket{\psi^{A_j}}, \text{ where } \ket{\psi^{A_j}} \in A_j \quad \forall j,\\
	\text{or, } &\ket{x}= \lambda\otimes_{j=1}^n \ket{\tilde{\psi}^{A_j}}, 
	\end{align*}
	where $\lambda = ||x||, \quad \ket{\tilde{\psi}^{A_j}} = \frac{\ket{\psi^{A_j}}}{||\psi^{A_j}||} \quad \forall j$.
	
	We can extend $\ket{\tilde{\psi}^{A_j}}$ into an orthonormal basis set $\{\ket{u_i^{A_j}}\}_{i=1}^{N_{A_j}} \subseteq A_j$ $(\ket{\tilde{\psi}^{A_j}} = \ket{u_1^{A_j}})$ by Gram-Schmidt orthonormalization process, for each $j$. Thus for each $j$, $A_j$ has an orthonormal basis set $\{\ket{u_i^{A_j}}\}_{i=1}^{N_{A_j}}$ such that $\ket{x}$ can be written as
	\begin{equation*}
	\ket{x}=\sum_{i=1}^m \lambda_i \otimes_{j=1}^n \ket{u_i^{A_j}},
	\end{equation*}
	where $m=min\{dim(A_j)|j=1,\dots,n\}$ and
	$$
	\lambda_i=
	\begin{cases}
	\lambda; \text{ if }i=1,\\
	0; \text{ otherwise.}
	\end{cases}
	$$
	Hence by Definition [\ref{shdn}] $\ket{x}$ has a Schmidt decomposition.
	
\end{proof}

\begin{lemma}\label{sh1iffpro}
	Let $A,B$ be finite dimensional Hilbert spaces. A pure state $0\neq\ket{x}\in A\otimes B$ has Schmidt number 1 (Def[\ref{sn}]), if and only if it is completely separable (Def[\ref{pss}]).
\end{lemma}

\begin{proof}
	$(\Rightarrow)$\\
	Let $\ket{x}$ have Schmidt number 1. This implies Schmidt decomposition of $\ket{x}$ has exactly 1 nonzero term. Thus $\ket{x}$ can be written as,
	\begin{equation*}
	\ket{x} = \lambda\ket{u^A}\ket{u^B},
	\end{equation*}
	where $\lambda=||x||$ and $\ket{u^A}\in A$, $\ket{u^B}\in B$ are normalized. Clearly $\ket{x}$ is a completely separable state by Definition[\ref{css}] since it can be written as,
	\begin{equation*}
	\ket{x} = \ket{{u^\prime}^A}\otimes\ket{u^B},
	\end{equation*}
	where $\ket{{u^\prime}^A}= \lambda\ket{u^A}\in A$ and $\ket{u^B}\in B$.\par
	$(\Leftarrow)$ Conversely, Let $\ket{x}$ be completely separable. From proof of Lemma[\ref{csl}] we observe that $\ket{x}$ can be written as,
	\begin{equation*}
	\ket{x} = \sum_{i=1}^{m}\lambda_i\ket{u_i^{A}}\ket{u_i^B},
	\end{equation*}
	where $\{u_i^A\}_{i=1}^m\subseteq A$ and $\{u_i^B\}_{i=1}^m\subseteq B$ are orthonormal sets, $m=min\{dim(A),dim(B)\}$ and
	$$
	\lambda_i
	\begin{cases}
	\neq 0; \text{ if }i=1,\\
	=0; \text{ otherwise.}
	\end{cases}
	$$
	Thus $\ket{x}$ has only 1 nonzero term in its Schmidt decomposition. Hence by Definition[\ref{sn}], $\ket{x}$ has Schmidt number 1.
	
\end{proof}

Let $A_1,A_2,A_3$ be finite dimensional Hilbert spaces with dimensions $N_{A_1},N_{A_2},N_{A_3}$ respectively. Let $\ket{x} \in \otimes_{j=1}^3 A_j$ be a pure state.

\begin{lemma}\label{cs2}
	Let us assume for each $j\in\{1,2,3\}$ $\exists$ orthonormal basis $\{\ket{u_i^{A_j}}\}_{i=1}^{N_{A_j}} \subseteq A_j$ such that
	$$
	_{A_j}\braket{u_i^{A_j}|x}=
	\begin{cases}
	\lambda_{ij} \otimes_{k\neq j}\ket{v_{ij}^{A_k}}; \text{ if } _{A_j}\braket{u_i^{A_j}|x}\neq0,\\
	0; \text{ otherwise,}
	\end{cases}
	$$
	where $\lambda_{ij} = ||_{A_j}\braket{u_i^{A_j}|x}||$ and $0 \neq \ket{v_{ij}^{A_k}} \in A_k$ is normalized for each $k \neq j$.
	\par 
	If $\exists$ some $s\in \{1,2,3\}$ such that for some $t \in \{i;\text{ }_{A_s}\braket{u_i^{A_s}|x}\neq0\}$ and $j\neq s$, we have $\braket{u_t^{A_s}|v_{ij}^{A_s}}\neq 0$ for at least 2 distinct values of $i$, then we can find an orthonormal basis $\{\ket{\tilde{u}_i^{A_j}}\}_{i=1}^{N_{A_j}} \subseteq A_j$ for which
	$$
	_{A_j}\braket{\tilde{u}_i^{A_j}|x}=
	\begin{cases}
	\tilde{\lambda}_{ij} \otimes_{k\neq j}\ket{\tilde{v}_{ij}^{A_k}}; \text{ if } _{A_j}\braket{\tilde{u}_i^{A_j}|x}\neq 0,\\
	0; \text{ otherwise,}
	\end{cases}
	$$
	and
	$$
	|\{i;\text{ }_{A_j}\braket{\tilde{u}_i^{A_j}|x}\neq 0\}| = |\{i;\text{ }_{A_j}\braket{u_i^{A_j}|x}\neq 0\}| - 1,
	$$
	where $\tilde{\lambda}_{ij} = ||_{A_j}\braket{\tilde{u}_i^{A_j}|x}||$ and $\ket{\tilde{v}_{ij}^{A_k}} \in A_k$ is normalized for each $k\neq j$. In other words, the new orthonormal basis has one less component with nonzero partial inner product with $\ket{x}$.
\end{lemma}

\begin{proof}
	Without loss of generality we can assume $j=1$ and $s=2$ and proceed with the proof. Clearly then we can write,
	\begin{equation*}
	\ket{x} = \sum_{i\in {\{a;_{A_1}\braket{u_a^{A_1}|x}\neq 0\}}} \lambda_{i1} \ket{u_i^{A_1}}\ket{v_{i1}^{A_2}}\ket{v_{i1}^{A_3}}.
	\end{equation*}
	Now let $\braket{u_t^{A_2}|v_{i1}^{A_2}} \neq 0$ for at least two distinct $i$, let $\tilde{i},\tilde{\tilde{i}}$, where $\tilde{i} \neq \tilde{\tilde{i}}$. Let us observe that,
	\begin{equation}\label{23}
	_{A_2}\braket{u_t^{A_2}|x} = \sum_{i\in{\{a;_{A_1}\braket{u_a^{A_1}|x}\neq 0\}}} \lambda_{i1} \braket{u_t^{A_2}|v_{i1}^{A_2}}\ket{u_i^{A_1}}\ket{v_{i1}^{A_3}}.
	\end{equation}
	We can always write $\ket{v_{i1}^{A_3}}$ in terms of the orthonormal basis $\{\ket{u_l^{A_3}}\}_{l=1}^{N_{A_3}}$,
	\begin{equation}\label{24}
	\ket{v_{i1}^{A_3}} = \sum_{l=1}^{N_{A_3}}\braket{u_l^{A_3}|v_{i1}^{A_3}}\ket{u_l^{A_3}}.
	\end{equation}
	Replacing $\ket{v_{i1}^{A_3}}$ in equation(\ref{23}) with the above representation(\ref{24}) we get,
	\begin{align}\label{25}
	&_{A_2}\braket{u_t^{A_2}|x} = \sum_{i} \lambda_{i1} \braket{u_t^{A_2}|v_{i1}^{A_2}}\ket{u_i^{A_1}}(\sum_{l=1}^{N_{A_3}}\braket{u_l^{A_3}|v_{i1}^{A_3}}\ket{u_l^{A_3}})\\
	&\text{ where } i\in{\{a;\text{ }_{A_1}\braket{u_a^{A_1}|x}\neq 0\}}.\nonumber
	\end{align}
	From the above equation(\ref{25}) we get,
	\begin{equation*}
	_{A_2}\braket{u_t^{A_2}|x} = \sum_{i}\ket{u_i^{A_1}}(\sum_{l=1}^{N_{A_3}} f_{il}\ket{u_l^{A_3}}),
	\end{equation*}
	where $f_{il} = \lambda_{i1} \braket{u_t^{A_2}|v_{i1}^{A_2}}\braket{u_l^{A_3}|v_{i1}^{A_3}}$.
	But $_{A_2}\braket{u_t^{A_2}|x}$ has Schmidt number 1 by lemma[\ref{sh1iffpro}], since $_{A_2}\braket{u_t^{A_2}|x}$ is completely separable in $A_1\otimes A_3$. Hence the matrix $[f_{il}]$ must have linearly dependent rows. This implies that the nonzero vectors $\sum_{l=1}^{N_{A_3}}f_{\tilde{i}l}\ket{u_l^{A_3}}$ and $\sum_{l=1}^{N_{A_3}}f_{\tilde{\tilde{i}}l}\ket{u_l^{A_3}}$ are linearly dependent. This further implies that the vectors $\ket{v_{\tilde{i}1}^{A_3}}$ and $\ket{v_{\tilde{\tilde{i}}1}^{A_3}}$ are actually linearly dependent, i.e.,
	\begin{equation}\label{ld}
	\ket{v_{\tilde{i}1}^{A_3}} = c_1\ket{v_{\tilde{\tilde{i}}1}^{A_3}},
	\end{equation}
	where $c_1 \neq 0 \text{ is constant}$.
	Now let us observe that $\exists$ at least one $t^\prime \in \{a;\text{ }_{A_3}\braket{u_a^{A_3}|x}\neq 0\}$ such that $\braket{u_{t^\prime}^{A_3}|{v_{\tilde{i}1}^{A_3}}} \neq 0$. Then by equation(\ref*{ld}) we get, $\braket{u_{t^\prime}^{A_3}|{v_{\tilde{\tilde{i}}1}^{A_3}}} \neq 0$ as well. In other words, $\braket{u_{t^\prime}^{A_3}|{v_{i1}^{A_3}}} \neq 0$ for at least 2 distinct values of $i=\tilde{i},\tilde{\tilde{i}}$, where $\tilde{i}\neq\tilde{\tilde{i}}$, i.e., the choice $s=2$ is not necessary for the proof; any arbitrary $s \neq 1$ could have been chosen. Then following the exactly same logic as before we can say that the vectors $\ket{v_{\tilde{i}1}^{A_2}}$ and $\ket{v_{\tilde{\tilde{i}}1}^{A_2}}$ are linearly dependent, i.e.,
	\begin{equation*}
	\ket{v_{\tilde{i}1}^{A_2}} = c_2\ket{v_{\tilde{\tilde{i}}1}^{A_2}},
	\end{equation*}
	where $c_2 \neq 0 \text{ is constant}$.
	Now we can observe that $\ket{x}$ may be written as,
	\begin{align*}
	\ket{x} = &\sum_{i\neq \tilde{i},\tilde{\tilde{i}}}\lambda_{i1} \ket{u_i^{A_1}}\ket{v_{i1}^{A_2}}\ket{v_{i1}^{A_3}} +\\ &(\lambda_{\tilde{i}1}c_1c_2\ket{u_{\tilde{i}}^{A_1}} + \lambda_{\tilde{\tilde{i}}1}\ket{u_{\tilde{\tilde{i}}}^{A_1}})\ket{v_{\tilde{\tilde{i}}1}^{A_2}}\ket{v_{\tilde{\tilde{i}}1}^{A_3}},
	\end{align*}
	where $i\in{\{a;_{A_1}\braket{u_a^{A_1}|x}\neq 0\}}$, i.e.,
	\begin{align*}
	\ket{x}= \sum_{i\neq \tilde{i},\tilde{\tilde{i}}}\lambda_{i1} \ket{u_i^{A_1}}\ket{v_{i1}^{A_2}}\ket{v_{i1}^{A_3}} + (a\ket{u_{\tilde{i}}^{A_1}} + b\ket{u_{\tilde{\tilde{i}}}^{A_1}})\ket{v_{\tilde{\tilde{i}}1}^{A_2}}\ket{v_{\tilde{\tilde{i}}1}^{A_3}},
	\end{align*}
	where $a=\lambda_{\tilde{i}1}c_1c_2$ and $b=\lambda_{\tilde{\tilde{i}}1}$.
	Now we construct an orthonormal basis $\{\ket{\tilde{u}_{i}^{A_1}}\}_{i=1}^{N_{A_1}}$ for $A_1$,
	$$
	\ket{\tilde{u}_{i}^{A_1}} = 
	\begin{cases}
	\ket{u_i^{A_1}}; \text{ if }i\neq \tilde{i},\tilde{\tilde{i}},\\
	\frac{a\ket{u_{\tilde{i}}^{A_1}} + b\ket{u_{\tilde{\tilde{i}}}^{A_1}}}{||a\ket{u_{\tilde{i}}^{A_1}} + b\ket{u_{\tilde{\tilde{i}}}^{A_1}}||};\text{ if } i=\tilde{i},\\
	\frac{b\ket{u_{\tilde{i}}^{A_1}} - a\ket{u_{\tilde{\tilde{i}}}^{A_1}}}{||b\ket{u_{\tilde{i}}^{A_1}} - a\ket{u_{\tilde{\tilde{i}}}^{A_1}}||};\text{ if } i=\tilde{\tilde{i}}.
	\end{cases}
	$$
	It is easy to observe that
	$$
	_{A_1}\braket{\tilde{u}_{i}^{A_1}|x} = 
	\begin{cases}
	_{A_1}\braket{u_{i}^{A_1}|x}; \text{ if }i\neq \tilde{i},\tilde{\tilde{i}},\\
	||a\ket{u_{\tilde{i}}^{A_1}} + b\ket{u_{\tilde{\tilde{i}}}^{A_1}}||\ket{v_{\tilde{\tilde{i}}1}^{A_2}}\ket{v_{\tilde{\tilde{i}}1}^{A_3}};\text{ if } i=\tilde{i},\\
	0; \text{ if }i=\tilde{\tilde{i}}.
	\end{cases}
	$$
	Thus $\{\ket{\tilde{u}_{i}^{A_1}}\}_{i=1}^{N_{A_1}}$ is an orthonormal basis for $A_1$ such that,
	$$
	_{A_1}\braket{\tilde{u}_i^{A_1}|x}=
	\begin{cases}
	\tilde{\lambda}_{i1} \otimes_{k\neq 1}\ket{\tilde{v}_{i1}^{A_k}}; \text{ if } _{A_1}\braket{\tilde{u}_i^{A_1}|x}\neq 0,\\
	0; \text{ otherwise,}
	\end{cases}
	$$
	and
	$$
	|\{a;\text{ }_{A_1}\braket{\tilde{u}_a^{A_1}|x}\neq 0\}|=|\{a;\text{ }_{A_1}\braket{u_a^{A_1}|x}\neq 0\}|-1,
	$$
	where $\tilde{\lambda}_{i1} = ||_{A_1}\braket{\tilde{u}_i^{A_1}|x}||$ and $\ket{\tilde{v}_{i1}^{A_k}} \in A_k$ is normalized for each $k\neq 1$. In other words, the new orthonormal basis $\{\ket{\tilde{u}_{i}^{A_1}}\}_{i=1}^{N_{A_1}}$ is the same as the old orthonormal basis $\{\ket{u_{i}^{A_1}}\}_{i=1}^{N_{A_1}}$, except for two terms; for indices $\tilde{i}$ and $\tilde{\tilde{i}}$. These two terms of the new orthonormal basis are actually suitable linear combinations of $\ket{u_{\tilde{i}}^{A_1}}$ and $\ket{u_{\tilde{\tilde{i}}}^{A_1}}$.\par Now for any other value of $j\neq1$ the exact same calculation will result as nowhere we have applied any specific property of Hilbert space $A_1$ within the proof. Thus our claim is proved for any arbitrary choice of $j$ and $s$.
	
\end{proof}

\subsection{Corrected version of the sufficient condition}
Now let us state a sufficient condition for the existence of Schmidt decomposition of a pure state in a tripartite Hilbert space. The following theorem is actually the corrected version of Proposition[\ref{wrth}].

\begin{theorem}\label{3p}
	Let $A_1,A_2,A_3$ be finite dimensional Hilbert spaces with dimensions $N_{A_1},N_{A_2},N_{A_3}$ respectively. Let $0\neq \ket{x}\in\otimes_{j=1}^{3}A_j$ be a pure state. Then $\ket{x}$ has Schmidt decomposition if $\exists$ orthonormal basis $\{\ket{u_i^{A_j}}\}_{i=1}^{N_{A_j}}$ for each $j$ such that if $_{A_j}\braket{u_i^{A_j}|x}\neq 0$ then $_{A_j}\braket{u_i^{A_j}|x}$ is completely separable in $\otimes_{k\neq j}A_k$.
\end{theorem}

\begin{proof}
	We provide only the major steps of the proof here. The detailed Proof can be found in Appendix~\ref{app1}.\\
	\textbf{Direction of proof:}\\
	We have for each $j\in \{1,2,3\}$,
	$$
	_{A_j}\braket{u_i^{A_j}|x}=
	\begin{cases}
	\lambda_{ij} \otimes_{k\neq j}\ket{v_{ij}^{A_k}}; \text{ if } _{A_j}\braket{u_i^{A_j}|x}\neq 0,\\
	0; \text{ otherwise,}
	\end{cases}
	$$
	where $\lambda_{ij} = ||_{A_j}\braket{u_i^{A_j}|x}||$, and $0 \neq \ket{v_{ij}^{A_k}} \in A_k$ is normalized for each $k \neq j$.\par
	Let us define for each $j\in\{1,2,3\}$,
	\begin{equation*}
	m_j := |\{i;\text{ }_{A_j}\bra{u_i^{A_j}}\ket{x}\neq 0\}|.
	\end{equation*}
	We also define,
	\begin{equation*}
	m := max\{m_1,m_2,m_3\}.
	\end{equation*}
	The result is proved by induction on $m$.
	\begin{enumerate}
		\item \textbf{Base case}: For $m=1$, $\ket{x}$ is completely separable. Hence proof is a direct consequence of Lemma[\ref{csl}].
		
		\item \textbf{Inductive hypotheses}: Theorem[\ref{3p}] assumed to be true for $m=\tilde{m}-1$.
		
		\item For $m=\tilde{m}$, two cases may arise.
		\begin{enumerate}
			\item \textbf{Case I:} $m_1=m_2=m_3=\tilde{m}$.\\
			Here two subcases may arise.
			\begin{enumerate}
				\item \textbf{Subcase I.A:} one-one correspondence between the sets $\{\ket{u_i^{A_k}};_{A_k}\braket{u_i^{A_k}|x}\neq 0\}$ and $\{\ket{v_{ij}^{A_k}}\}$ for some $j$ and for each $k\neq j$.
				\item \textbf{Subcase I.B:} For each $j$ $\exists s\neq j$ such that for some $t\in \{i;\text{ }_{A_s}\braket{u_i^{A_s}|x}\neq 0\}$, $\braket{u_t^{A_s}|v_{ij}^{A_s}}\neq 0$ for at least 2 distinct values of $i$. Then the use of Lemma[\ref{cs2}] suggests that the value of $m$ can be reduced by 1. Thus proof comes directly from inductive hypotheses.
			\end{enumerate}
			
			\item \textbf{Case II:} $m_1,m_2,m_3$ not all equal.\\
			In this case also the use of Lemma[\ref{cs2}] suggests that the value of $m$ can be reduced by 1. Thus proof comes directly from inductive hypotheses.
		\end{enumerate}
	\end{enumerate}
	Hence the proof is complete.
	
\end{proof}

\section{Extension of the sufficient condition to multipartite Hilbert spaces}
Now we have a sufficient condition for existence of Schmidt decomposition of a pure state in a tripartite Hilbert space. The next theorem extends this result for an $n$-partite Hilbert space (Def[\ref{nph}]).

\begin{theorem}\label{np}
	Let $A_1,A_2,\dots,A_n$ $(n\geq 3)$ be finite dimensional Hilbert spaces with dimensions $N_{A_1},N_{A_2},\dots,N_{A_n}$ respectively. Then a nonzero pure state $\ket{x} \in \otimes_{j=1}^n A_j$ has Schmidt decomposition if $\exists$ orthonormal basis $\{\ket{u_i^{A_j}}\}_{i=1}^{N_{A_j}} \subseteq A_j$ for each $j$, such that, if $_{A_j}\braket{u_i^{A_j}|x}\neq 0$ then $_{A_j}\braket{u_i^{A_j}|x}$ is completely separable in $\otimes_{k\neq j}A_k$.
\end{theorem}

\begin{proof}
	We will prove this by induction on the number of parts $n$.\\
	\textbf{Base of induction:} $n=3:$\\
	The claim holds for a tripartite Hilbert space is shown in Theorem[\ref{3p}].\\
	\textbf{Inductive hypothesis:} Let us assume that the claim is true for $(\tilde{n}-1)$-partite Hilbert spaces.\\
	\textbf{Proof for $n=\tilde{n}$:}\\
	Let us look at $\ket{x}$ as a state in an $(\tilde{n}-1)$-partite Hilbert Space by observing the bipartite Hilbert space $A_{\tilde{n}-1}\otimes A_{\tilde{n}}$ as a single finite dimensional Hilbert space $B$. Now by assumption $\ket{x}$ can be written as,
	\begin{equation*}
	\ket{x} = \sum_{i=1}^{N_{A_{\tilde{n}}}}\lambda_{i\tilde{n}}\ket{v_{i\tilde{n}}^{A_1}}\ket{v_{i\tilde{n}}^{A_2}}\dots\big{(}\ket{v_{i\tilde{n}}^{A_{\tilde{n}-1}}}\ket{u_i^{A_{\tilde{n}}}}\big{)},
	\end{equation*}
	where $\lambda_{i\tilde{n}} = ||_{A_{\tilde{n}}}\braket{u_i^{\tilde{n}}|x}||$ and $\ket{v_{i\tilde{n}}^{A_j}}\in A_j$ is normalized for each $j$. Careful observation shows that the set $\{\ket{v_{i\tilde{n}}^{A_{\tilde{n}-1}}}\ket{u_i^{A_{\tilde{n}}}}\}_{i=1}^{N_{A_{\tilde{n}}}}$ is orthonormal in the Hilbert space $B$. Hence it can be extended to an orthonormal basis $\{\ket{\beta_{i}}\}_{i=1}^{(N_{A_{\tilde{n}-1}}N_{A_{\tilde{n}}})}$ of $B$ by Gram-Schmidt orthonormalization. Now let us observe that $\exists$ orthonormal basis $\{\ket{u_i^{A_j}}\}_{i=1}^{N_{A_j}} \subseteq A_j$, $\forall j\in\{1,\dots,\tilde{n}-2\}$, and $\{\ket{\beta_{i}}\}_{i=1}^{(N_{A_{\tilde{n}-1}}N_{A_{\tilde{n}}})}\subseteq B$ such that if $_{A_j}\braket{u_i^{A_j}|x}\neq 0$ $(_{B}\braket{\beta_i|x}\neq 0)$ then $_{A_j}\braket{u_i^{A_j}|x}$ $(_{B}\braket{\beta_i|x})$ is completely separable (Def[\ref{css}]) in $\mathop{\otimes_{k=1}^{\tilde{n}-2}}_{,k\neq j}A_k\otimes B$ $(\otimes_{k=1}^{\tilde{n}-2}A_k)$. By induction hypothesis this implies that $\ket{x}$ has Schmidt decomposition in the Hilbert space $\otimes_{k=1}^{\tilde{n}-2}A_k\otimes B$. Then $\ket{x}$ can be written as
	\begin{equation*}
	\ket{x} = \sum_{i=1}^m \lambda_i\ket{i_{A_1}}\dots\ket{i_{A_{\tilde{n}-2}}}\ket{i_B},\quad(\text{ Def}[\ref{shdn}])
	\end{equation*}
	where $\lambda_i \in \mathbb{C}$ $\forall i$, $m=min\{N_{A_1},\dots,N_{A_{\tilde{n}-2}},N_{A_{\tilde{n}-1}}N_{A_{\tilde{n}}}\}$ and $\{\ket{i_{A_j}}\}_{i=1}^{N_{A_j}}$ $(\{i_B\}_{i=1}^{(N_{A_{\tilde{n}-1}}N_{A_{\tilde{n}}})})$ is an orthonormal set in $A_j$ $(B)$ by Definition[\ref{shdn}] of Schmidt decomposition.
	\par 
	Now we claim that for each $i\in \{1,\dots,N_{A_{\tilde{n}}}\}$, $_{A_{\tilde{n}}}\braket{u_i^{A_{\tilde{n}}}|j_B}\neq 0$ for at most one $j\in\{1,\dots,N_{A_{\tilde{n}-1}}N_{A_{\tilde{n}}}\}$, if $_B\braket{j_B|x}\neq 0$. Suppose not, i.e. for some $i$, $_{A_{\tilde{n}}}\braket{u_i^{A_{\tilde{n}}}|j_B}\neq 0$ for at least 2 distinct values of $j$, let $j_1,j_2$, where $j_1\neq j_2$ and $_B\braket{j_B|x}\neq 0$ for $j=j_1,j_2$. Then we can write,
	\begin{align*}
	_{A_{\tilde{n}}}\braket{u_i^{A_{\tilde{n}}}|x} = &\sum_{j= j_1,j_2}\lambda_j\ket{j_{A_1}}\dots\ket{{j}_{A_{\tilde{n}-2}}}\big{(}_{A_{\tilde{n}}}\braket{u_i^{A_{\tilde{n}}}|{j}_B}\big{)} + \\
	&\sum_{j\neq j_1,j_2}\lambda_j\ket{j_{A_1}}\dots\ket{{j}_{A_{\tilde{n}-2}}}\big{(}_{A_{\tilde{n}}}\braket{u_i^{A_{\tilde{n}}}|{j}_B}\big{)}.
	\end{align*}
	Clearly $\lambda_{j_1},\lambda_{j_2}\neq 0$, since $_B\braket{j_B|x}\neq 0$ for $j=j_1,j_2$. Also $\{_{A_{\tilde{n}}}\braket{u_i^{A_{\tilde{n}}}|{j}_B}\}_j$ is an orthogonal set in $A_{\tilde{n}-1}$ since $\{\ket{j_B}\}_j$ is an orthonormal set in $B$. Now we can see $_{A_{\tilde{n}}}\braket{u_i^{A_{\tilde{n}}}|x}$ as a state in the bipartite system $A\otimes A_{\tilde{n}-1}$, where $A=\otimes_{k=1}^{\tilde{n}-2}A_k$. Then we can notice that,
	\begin{equation*}
	_{A_{\tilde{n}}}\braket{u_i^{A_{\tilde{n}}}|x} = \sum_{j}\lambda_j\ket{j_A}\big{(}_{A_{\tilde{n}}}\braket{u_i^{A_{\tilde{n}}}|{j}_B}\big{)},
	\end{equation*}
	where $\ket{j_A} = \ket{j_{A_1}}\dots\ket{{j}_{A_{\tilde{n}-2}}}$. This	is actually Schmidt decomposition of $_{A_{\tilde{n}}}\braket{u_i^{A_{\tilde{n}}}|x}$ in bipartite system $A\otimes A_{\tilde{n}-1}$ since $\{\ket{j_A}\}_j$ is an orthonormal set in $A$ and $\{_{A_{\tilde{n}}}\braket{u_i^{A_{\tilde{n}}}|{j}_B}\}_j$ is an orthogonal set in $A_{\tilde{n}-1}$. Clearly $_{A_{\tilde{n}}}\braket{u_i^{A_{\tilde{n}}}|x}\in A\otimes A_{\tilde{n}-1}$ has Schmidt number at least 2 since its Schmidt decomposition has at least 2 nonzero terms, for indices $j_1$ and $j_2$. This implies, by Lemma[\ref{sh1iffpro}], that $_{A_{\tilde{n}}}\braket{u_i^{A_{\tilde{n}}}|x}$ is not a separable state in $A\otimes A_{\tilde{n}-1}$, which further implies that it is not completely separable in $\otimes_{k=1}^{\tilde{n}-1}A_k$ and this is clearly a contradiction.
	Thus our claim is proved.\par
	Similarly we can also show that for each $i\in \{1,\dots,N_{A_{\tilde{n}-1}}\}$, $_{A_{\tilde{n}-1}}\braket{u_i^{A_{\tilde{n}-1}}|j_B}\neq 0$ for at most one $j\in\{1,\dots,N_{A_{\tilde{n}-1}}N_{A_{\tilde{n}}}\}$, provided $_B\braket{j_B|x}\neq 0$. \par
	Now let us define relations $R_1,R_2$,
	\begin{eqnarray*}
	R_1\subseteq \{1,\dots,N_{A_{\tilde{n}-1}}\}\times\{1,\dots,N_{A_{\tilde{n}-1}}N_{A_{\tilde{n}}}\}\nonumber\\
	(i,j)\in R_1; \text{ iff } (_{A_{\tilde{n}-1}}\braket{u_i^{A_{\tilde{n}-1}}|j_B}\neq 0 \text{ and } _B\braket{j_B|x}\neq 0),\\
	R_2\subseteq \{1,\dots,N_{A_{\tilde{n}}}\}\times\{1,\dots,N_{A_{\tilde{n}-1}}N_{A_{\tilde{n}}}\}\nonumber\\
	(i,j)\in R_2; \text{ iff } (_{A_{\tilde{n}}}\braket{u_i^{A_{\tilde{n}}}|j_B}\neq 0 \text{ and } _B\braket{j_B|x}\neq 0).
	\end{eqnarray*}
	From previous arguments we can see that for each $j$ such that $_B\braket{j_B|x}\neq 0$ $\exists$ exactly one $i$ for which $(i,j)\in R_1$ $(R_2)$, because, since $\{\ket{u_i^{A_{\tilde{n}-1}}}\}_i\text{ }(\{\ket{u_i^{A_{\tilde{n}}}}\}_i)$ is orthonormal basis for $A_{\tilde{n}-1}\text{ }(A_{\tilde{n}})$ and $\ket{j_B}\in A_{\tilde{n}-1}\otimes A_{\tilde{n}}$, $\exists$ at least one $i$ such that $_{A_{\tilde{n}-1}}\braket{u_i^{A_{\tilde{n}-1}}|j_B}\neq 0\text{ } (_{A_{\tilde{n}}}\braket{u_i^{A_{\tilde{n}}}|j_B}\neq 0)$. Thus neglecting all $i$ for which $\lambda_i=0$, we can write,
	\begin{equation*}
	\ket{x}=\sum_{i}\lambda_i\ket{i_{A_1}}\dots\ket{i_{A_{\tilde{n}-2}}}\ket{u_{R_1^{-1}(i)}^{A_{\tilde{n}-1}}}\ket{u_{R_2^{-1}(i)}^{A_{\tilde{n}}}},
	\end{equation*}
	where $\lambda_{i}\in \mathbb{C}\smallsetminus \{0\} \quad \forall i$.
	This is clearly Schmidt decomposition of $\ket{x}$ seen as a state of the Hilbert space $\otimes_{k=1}^{\tilde{n}}A_k$.\par
	Thus our claim is proved for $n=\tilde{n}$. Hence by induction, the claim holds $\forall n\in\mathbb{N}$, i.e., the Theorem[\ref{np}] is true for any $n$-partite Hilbert space.
	
\end{proof}

\section{Both necessary and sufficient condition for multipartite Hilbert spaces}
Finally we will observe that the sufficient condition is also a necessary condition.
\begin{theorem}\label{mainth}
	Let $A_1,A_2,\dots,A_n$ $(n\geq 3)$ be finite dimensional Hilbert spaces with dimensions $N_{A_1},N_{A_2},\dots,N_{A_n}$ respectively. A nonzero pure state $\ket{x} \in \otimes_{j=1}^n A_j$ has Schmidt decomposition if and only if $\exists$ orthonormal basis $\{\ket{u_i^{A_j}}\}_{i=1}^{N_{A_j}} \subseteq A_j$ for each $j$ such that, if $_{A_j}\braket{u_i^{A_j}|x}\neq 0$ then $_{A_j}\braket{u_i^{A_j}|x}$ is completely separable in $\otimes_{k\neq j}A_k$.
\end{theorem}

\begin{proof}
	$(\Rightarrow)$\\
	If $\exists$ orthonormal basis $\{\ket{u_i^{A_j}}\}_{i=1}^{N_{A_j}} \subseteq A_j$ for each $j$ such that, if $_{A_j}\braket{u_i^{A_j}|x}\neq 0$ then $_{A_j}\braket{u_i^{A_j}|x}$ is completely separable in $\otimes_{k\neq j}A_k$; then $\ket{x}$ has Schmidt decomposition. This has been proved in Theorem[\ref{np}].\par
	$(\Leftarrow)$
	Conversely, if $\ket{x}$ has Schmidt decomposition then it can be written as,
	\begin{equation*}
	\ket{x} = \sum_{i=1}^m \lambda_i \otimes_{j=1}^n\ket{i_{A_j}},
	\end{equation*}
	where $\lambda_i \in \mathbb{C} \quad \forall i$, $m=min\{N_{A_1},\dots,N_{A_n}\}$ and $\{\ket{i_{A_j}}\}_{i=1}^m$ is an orthonormal set in $A_j$, $\forall j\in\{1,\dots,n\}$. For each $A_j$ then, we can extend the set $\{\ket{i_{A_j}}\}_{i=1}^m$ to an orthonormal basis $\{\ket{i_{A_j}}\}_{i=1}^{N_{A_j}}$ by Gram-Schmidt orthonormalization. Clearly,
	$$
	_{A_j}\braket{i_{A_j}|x}=
	\begin{cases}
	\lambda_i \otimes_{k\neq j}\ket{i_{A_k}}; \text{ if }j\leq m,\\
	0; \text{ if }m<j\leq N_{A_j},
	\end{cases}
	$$
	i.e., if $_{A_j}\braket{i_{A_j}|x}\neq 0$ then $_{A_j}\braket{i_{A_j}|x}$ is completely separable in $\otimes_{k\neq j}A_k$. Thus the converse is also proved.\par
	Hence the proof of the theorem is completed.
	
\end{proof}

A simple but important corollary of the above Theorem[\ref{mainth}] is given below. The importance of this corollary is that, it has used to disprove the Proposition[\ref{wrth}] suggested by Dr. A. K. Pati \cite{1}.
\begin{corollary}\label{psc}
	Let $A_1,A_2,\dots,A_n$ $(n\geq 3)$ be finite dimensional Hilbert spaces with dimensions $N_{A_1},N_{A_2},\dots,N_{A_n}$ respectively. If a nonzero pure state $\ket{x}\in \otimes_{j=1}^n A_j$ is partially separable (Def[\ref{pss}]) but not completely separable (Def[\ref{css}]) then $\ket{x}$ does not have Schmidt decomposition.
\end{corollary}

\begin{proof}
	We will prove this by induction on the number of parts $n$.\\
	\textbf{Base of induction:} $n=3:$\\
	As $\ket{x}$ is partially separable but not completely separable, it can be written (without loss of generality) as
	\begin{equation*}
	\ket{x} = \ket{\psi^{A_1A_2}}\ket{\phi^{A_3}},
	\end{equation*}
	where $\ket{\psi^{A_1A_2}}\in A_1\otimes A_2$ is not separable and $\ket{\phi^{A_3}}\in A_3$. For any orthonormal basis $\{\ket{i_{A_3}}\}_{i=1}^{N_{A_3}}\subseteq A_3$,
	\begin{equation*}
	_{A_3}\braket{i_{A_3}|x} = \braket{i_{A_3}|\phi^{A_3}}\ket{\psi^{A_1A_2}}
	\end{equation*}
	is clearly not separable in $A_1\otimes A_2$, since $\ket{\psi^{A_1A_2}}$ is not separable in $A_1\otimes A_2$. Thus by Theorem[\ref{mainth}] $\ket{x}$ does not have Schmidt decomposition. Hence the claim holds for $n=3$.\\
	\textbf{Inductive hypothesis:} Let us assume that the claim is true $\forall n\leq\tilde{n}-1$.\\
	\textbf{Proof for $n=\tilde{n}$:}\\
	Let us view $\ket{x}$ as a state in an $(\tilde{n}-1)$-partite Hilbert space by assuming the bipartite Hilbert space $A_{\tilde{n}-1}\otimes A_{\tilde{n}}$ to be a single Hilbert space $B$. Now given $\ket{x}$ is partially separable but not completely separable, two cases may arise.
	\begin{enumerate}
		\item \textbf{Case I:} $\ket{x}$ is partially separable as a state in $\otimes_{j=1}^{\tilde{n}-2}A_j\otimes B$.\\
		Then by inductive hypothesis we can say that $\ket{x}$ does not have Schmidt decomposition. Thus the claim is proved for Case-I.
		
		\item \textbf{Case II:} $\ket{x}$ is not partially separable as a state in $\otimes_{j=1}^{\tilde{n}-2}A_j\otimes B$.\\
		Then $\ket{x}$ must be of the form
		\begin{equation*}
		\ket{x} = \ket{\psi^{A_1A_2\dots A_{\tilde{n}-1}}}\ket{\phi^{A_{\tilde{n}}}},
		\end{equation*}
		where $\ket{\psi^{A_1A_2\dots A_{\tilde{n}-1}}}\in\otimes_{j=1}^{\tilde{n}-1}A_j$ is not partially separable, because otherwise, $\ket{x}$ will be partially separable as a state of $\otimes_{j=1}^{\tilde{n}-2}A_j\otimes B$. Also, $\ket{\phi^{A_{\tilde{n}}}}\in A_{\tilde{n}}$. Now for any orthonormal basis $\{\ket{i_{A_{\tilde{n}}}}\}_{i=1}^{N_{A_{\tilde{n}}}}$ of $A_{\tilde{n}}$, we can write
		\begin{equation*}
		_{A_{\tilde{n}}}\braket{i_{A_{\tilde{n}}}|x} = \braket{i_{A_{\tilde{n}}}|\phi^{A_{\tilde{n}}}}\ket{\psi^{A_1A_2\dots A_{\tilde{n}-1}}}.
		\end{equation*}
		But this means, $_{A_{\tilde{n}}}\braket{i_{A_{\tilde{n}}}|x}$ not completely separable, since $\ket{\psi^{A_1A_2\dots A_{\tilde{n}-1}}} \in \otimes_{j=1}^{\tilde{n}-1}A_j$ is not completely separable. Hence by Theorem[\ref{mainth}] $\ket{x}$ does not have Schmidt decomposition. Thus the claim is proved for Case-II.
	\end{enumerate}
	Hence the claim is proved for $n=\tilde{n}$. By induction we can say that the claim is true $\forall n\in\mathbb{N}$. This completes the proof of the Corollary[\ref{psc}].
	
\end{proof}

\subsection{Applicability of the theorem for a bipartite Hilbert space}
As a sanity check, we show here how Theorem[\ref{mainth}] holds for bipartite systems as well. We will show that Theorem[\ref{mainth}] actually implies that every pure state in a bipartite Hilbert space is actually Schmidt decomposable.

\begin{corollary}\label{support}
	Every pure state in a bipartite Hilbert space has a Schmidt decomposition.
\end{corollary}

\begin{proof}
	Let $\ket{x}\in A\otimes B$ be a pure state, where $A$ and $B$ are finite dimensional Hilbert spaces with dimensions $N_A$ and $N_B$ respectively. Let $\{\ket{i_A}\}_{i=1}^{N_A}$ be an orthonormal basis in $A$ and $\{\ket{i_B}\}_{i=1}^{N_B}$ be an orthonormal basis in $B$. Then $\ket{x}$ can be represented in terms of $\{\ket{i_A}\}_{i=1}^{N_A}$ and $\{\ket{i_B}\}_{i=1}^{N_B}$ as,
	\begin{equation*}
	\ket{x} = \sum_{i=1}^{N_A}\sum_{j=1}^{N_B} \lambda_{ij}\ket{i_A}\ket{j_B},
	\end{equation*}
	where $\lambda_{ij}\in \mathbb{C}$ $\forall i,j$. Now we can easily observe that,
	\begin{align*}
	&_A\braket{i_A|x} = \sum_{j=1}^{N_B} \lambda_{ij}\ket{j_B} \quad \forall i\in\{1,\dots,N_A\},\\
	&_B\braket{j_B|x} = \sum_{i=1}^{N_A} \lambda_{ij}\ket{i_A} \quad \forall j\in\{1,\dots,N_B\}.
	\end{align*}
	But both $\sum_{j=1}^{N_B} \lambda_{ij}\ket{j_B}$ and $\sum_{i=1}^{N_A} \lambda_{ij}\ket{i_A}$ are completely separable states (if not $0$), since each of them belong to a single Hilbert space, $\sum_{j=1}^{N_B} \lambda_{ij}\ket{j_B}\in B$ and $\sum_{i=1}^{N_A} \lambda_{ij}\ket{i_A}\in A$. Thus by Theorem[\ref{mainth}], $\ket{x}$ has Schmidt decomposition. Since $\ket{x}$ was chosen arbitrarily, the result is true for any pure state $\ket{x}\in A\otimes B$ in general.
	
\end{proof}
We already know that any pure state in a bipartite Hilbert space has a Schmidt decomposition (Def[\ref{shdb}]). We see that Corollary[\ref{support}] is actually supporting that result. Thus Theorem[\ref{mainth}] remains valid for bipartite systems as well.

\appendix
\section{Detailed Proof of Theorem[\ref{3p}]}
\label{app1}
\begin{proof}
	We have for each $j\in \{1,2,3\}$,
	$$
	_{A_j}\braket{u_i^{A_j}|x}=
	\begin{cases}
	\lambda_{ij} \otimes_{k\neq j}\ket{v_{ij}^{A_k}}; \text{ if } _{A_j}\braket{u_i^{A_j}|x}\neq 0,\\
	0; \text{ otherwise,}
	\end{cases}
	$$
	where $\lambda_{ij} = ||_{A_j}\braket{u_i^{A_j}|x}||$ and $0 \neq \ket{v_{ij}^{A_k}} \in A_k$ is normalized for each $k \neq j$.
	\par
	Let $m_j := |\{i;\text{ }_{A_j}\braket{u_i^{A_j}|x}\neq 0\}|$ be defined $\forall j\in\{1,2,3\}$. We also define,
	\begin{equation*}
	m := max\{m_1,m_2,m_3\}.
	\end{equation*}
	If we can prove that the theorem is true $\forall m\in \mathbb{N}$ then we are done. We will proceed by induction on $m$.\\
	\textbf{Base of induction:} $m=1$.\\
	Then $|\{i;\text{ }_{A_j}\braket{u_i^{A_j}|x}\neq 0\}| = 1$ $\forall j\in\{1,2,3\}$, since $|\{i;\text{ }_{A_j}\braket{u_i^{A_j}|x}\neq 0\}|\neq0$ as $\ket{x}\neq 0$. Then $\ket{x}$ must be of the form,
	\begin{equation*}
	\ket{x} = \lambda\otimes_{j=1}^3 \ket{u_{i_j}^{A_j}},
	\end{equation*}
	where $\lambda = \lambda_{i_11} = \lambda_{i_22} = \lambda_{i_33}$ and $i_j\in\{i;\text{ }_{A_j}\braket{u_i^{A_j}|x}\neq 0\}$ $\forall j\in\{1,2,3\}$. This is because, for each $j$,
	\begin{align*}
	& _{A_j}\braket{u_i^{A_j}|x} = 0 \quad \forall i\neq i_j\\
	\implies& \ket{v_{i_jj}^{A_k}} = \ket{u_{i_k}^{A_k}} \quad \forall k\neq j
	\end{align*}
	Thus $\ket{x}$ is completely separable (Def[\ref{css}]). Hence by Lemma[\ref{csl}] $\ket{x}$ has Schmidt decomposition.\\
	\textbf{Inductive hypothesis:} Let us assume that the theorem is true $\forall m\leq \tilde{m}-1$.\\
	\textbf{Proof for $m=\tilde{m}$:} Two cases may arise.
	\begin{enumerate}
		\item \textbf{Case I:} $m_1=m_2=m_3 = \tilde{m}$.\\
		Then $\ket{x}$ can be written as,
		\begin{align*}
		\ket{x} &= \sum_{i\in\{a;_{A_1}\braket{u_a^{A_1}|x}\neq0\}} \lambda_{i1}\ket{u_i^{A_1}}\ket{v_{i1}^{A_2}}\ket{v_{i1}^{A_3}}\\
		&= \sum_{i\in\{a;_{A_2}\braket{u_a^{A_2}|x}\neq0\}} \lambda_{i2}\ket{v_{i2}^{A_1}}\ket{u_i^{A_2}}\ket{v_{i2}^{A_3}}\\
		&= \sum_{i\in\{a;_{A_3}\braket{u_a^{A_3}|x}\neq0\}} \lambda_{i3}\ket{v_{i3}^{A_1}}\ket{v_{i3}^{A_2}}\ket{u_i^{A_3}}.
		\end{align*}
		Now two subcases may arise,
		\begin{enumerate}
			\item \textbf{Subcase I.A:} For some $j\in \{1,2,3\}$ $\exists$ a 1-1 correspondence between the sets $\{\ket{v_{ij}^{A_k}};\text{ }i\in\{a;_{A_j}\braket{u_a^{A_j}|x}\neq0\}\}$ and $\{\ket{u_i^{A_k}};\text{ }i\in\{a;_{A_k}\braket{u_a^{A_k}|x}\neq0\}\}$ $\forall k\neq j$.\\
			In other words, $\exists$ a one-one map $f^k: \{a;_{A_j}\braket{u_a^{A_j}|x}\neq0\} \longrightarrow\{a;_{A_k}\braket{u_a^{A_k}|x}\neq0\}$ $\forall k\neq j$ such that,
			\begin{equation*}
			\ket{v_{ij}^{A_k}} = \ket{u_{f^k(i)}^{A_k}}.
			\end{equation*}
			Then clearly we can write,
			\begin{align*}
			\ket{x} &= \sum_{i\in \{a;_{A_j}\braket{u_a^{A_j}|x}\neq0\}} \lambda_{ij}\ket{u_i^{A_j}}\otimes_{k\neq j}\ket{v_{ij}^{A_k}}\\
			&= \sum_{i\in \{a;_{A_j}\braket{u_a^{A_j}|x}\neq0\}} \lambda_{ij}\ket{u_i^{A_j}}\otimes_{k\neq j}\ket{u_{f^k(i)}^{A_k}}.
			\end{align*}
			Now this is a Schmidt decomposition of $\ket{x}$ by Def[\ref{shdn}] since $\{u_{f^k(i)}^{A_k}\}$ is an orthonormal set in $A_k$ for each $k\neq j$. This is because $f^k$ is one-one and $i_1\neq i_2$ implies $f^k(i_1)\neq f^k(i_2)$. Thus $\ket{x}$ has Schmidt decomposition and our claim holds for Subcase-I.A.
			
			\item \textbf{Subcase I.B:} $\nexists j\in\{1,2,3\}$ such that a 1-1 correspondence exists between the sets $\{\ket{v_{ij}^{A_k}};\text{ }i\in\{a;_{A_j}\braket{u_a^{A_j}|x}\neq0\}\}$ and $\{\ket{u_i^{A_k}};\text{ }i\in\{a;_{A_k}\braket{u_a^{A_k}|x}\neq0\}\}$ $\forall k\neq j$.\\
			Then for each $j$ $\exists$ at least one $k\neq j$ and at least one $i\in \{a;_{A_j}\braket{u_a^{A_j}|x}\neq0\}$ such that $\ket{v_{ij}^{A_k}}$ has nonzero inner product with $\ket{u_{i^\prime}^{A_k}}$ for at least 2 distinct values of $i^\prime \in \{a;_{A_k}\braket{u_a^{A_k}|x}\neq0\}$. Since $|\{a;_{A_j}\braket{u_a^{A_j}|x}\neq0\}| = |\{a;_{A_k}\braket{u_a^{A_k}|x}\neq0\}|$, by pigeon hole principle we can then say that $\exists$ $t \in \{a;_{A_k}\braket{u_a^{A_k}|x}\neq0\}$ such that $\braket{u_t^{A_k}|v_{ij}^{A_k}}\neq 0$ for at least 2 distinct values of $i\in \{a;_{A_j}\braket{u_a^{A_j}|x}\neq0\}$. This implies by Lemma[\ref{cs2}] that for each $A_j$ $\exists$ orthonormal basis $\{\ket{\tilde{u}_i^{A_j}}\}_{i=1}^{N_{A_j}}$ such that,
			$$
			_{A_j}\braket{\tilde{u}_i^{A_j}|x}=
			\begin{cases}
			\tilde{\lambda}_{ij}\otimes_{k\neq j}\ket{\tilde{v}_{ij}^{A_k}}; \text{ if } _{A_j}\braket{\tilde{u}_i^{A_j}|x}\neq0,\\
			0; \text{ otherwise,}
			\end{cases}
			$$
			and
			\begin{equation*}
			|\{i;_{A_j}\braket{\tilde{u}_i^{A_j}|x}\neq 0\}| = |\{i;_{A_j}\braket{u_i^{A_j}|x}\neq 0\}|-1
			\end{equation*}
			In other words, for each $j$ we can find a new orthonormal basis $\{\ket{\tilde{u}_i^{A_j}}\}_{i=1}^{N_{A_j}}$ for $A_j$ which preserves the condition stated in the Theorem[\ref{3p}] but $\ket{x}$ has partial inner product zero with one more component than $\{\ket{u_i}^{A_j}\}_{i=1}^{N_{A_j}}$.
			Hence, for the given state $\ket{x}\in \otimes_{j=1}^3 A_j$ we can find orthonormal basis $\{\ket{\tilde{u}_i^{A_j}}\}_{i=1}^{N_{A_j}}\subseteq A_j$ $\forall j$ such that,
			$$
			_{A_j}\braket{\tilde{u}_i^{A_j}|x}=
			\begin{cases}
			\tilde{\lambda}_{ij}\otimes_{k\neq j}\ket{\tilde{v}_{ij}^{A_k}}; \text{ if } _{A_j}\braket{\tilde{u}_i^{A_j}|x}\neq0,\\
			0; \text{ otherwise,}
			\end{cases}
			$$
			and
			\begin{equation*}
			|\{i;_{A_j}\braket{\tilde{u}_i^{A_j}|x}\neq 0\}| = m_j-1,
			\end{equation*}
			i.e., $m=max\{m_1-1,m_2-1,m_3-1\}=\tilde{m}-1$. Using induction hypothesis we can say that $\ket{x}$ has Schmidt decomposition. Thus the claim holds for Subcase-I.B.
		\end{enumerate}
		Thus we have shown that the theorem is true for Case-I.
		
		\item \textbf{Case II:} $m_1,m_2,m_3$ are not all equal.\\
		Without loss of generality, let $m_1 = max\{m_1,m_2,m_3\} = \tilde{m}$ and $m_3<m_1$. By assumption we can write,
		\begin{equation*}
		\ket{x}=\sum_{i\in \{a;_{A_1}\braket{u_a^{A_1}|x}\neq 0\}} \lambda_{i1}\ket{u_i^{A_1}}\ket{v_{i1}^{A_2}}\ket{v_{i1}^{A_3}}.
		\end{equation*}
		This implies $|\{\ket{v_{i1}^{A_3}};\text{ }i\in \{a;_{A_1}\braket{u_a^{A_1}|x}\neq 0\}\}| = \tilde{m} > m_3$. But we know that,
		\begin{equation*}
		|\{a;_{A_3}\braket{u_a^{A_3}|x}\neq 0\}| = m_3.
		\end{equation*}
		This implies, by Pigeonhole principle, $\exists$ at least one $t \in \{a;_{A_3}\braket{u_a^{A_3}|x}\neq 0\}$ such that $\braket{u_t^{A_3}|v_{i1}^{A_3}}\neq 0$ for at least 2 distinct values of $i\in \{a;_{A_1}\braket{u_a^{A_1}|x}\neq 0\}$. Using Lemma[\ref{cs2}] we get that $\exists$ orthonormal basis $\{\ket{\tilde{u}_i^{A_1}}\}_{i=1}^{N_{A_1}} \subseteq A_1$ such that
		$$
		_{A_1}\braket{u_i^{A_1}|x} = 
		\begin{cases}
		\tilde{\lambda}_{i1}\ket{\tilde{v}_{i1}^{A_2}}\ket{\tilde{v}_{i1}^{A_3}}; \text{ if } _{A_1}\braket{u_i^{A_1}|x}\neq 0,\\
		0; \text{ otherwise,}
		\end{cases}
		$$
		and
		\begin{equation*}
		|\{i;_{A_1}\braket{\tilde{u}_i^{A_1}|x}\neq 0\}| = m_1-1.
		\end{equation*}
		If $m_2 = max\{m_1,m_2,m_3\} = \tilde{m}$ as well then we can prove similarly that $\exists$ orthonormal basis $\{\ket{\tilde{u}_i^{A_2}}\}_{i=1}^{N_{A_2}} \subseteq A_2$ such that,
		$$
		_{A_2}\braket{u_i^{A_2}|x} = 
		\begin{cases}
		\tilde{\lambda}_{i2}\ket{\tilde{v}_{i2}^{A_1}}\ket{\tilde{v}_{i2}^{A_3}}; \text{ if }_{A_2}\braket{u_i^{A_2}|x}\neq 0,\\
		0; \text{ otherwise,}
		\end{cases}
		$$
		and
		\begin{equation*}
		|\{i;_{A_2}\braket{\tilde{u}_i^{A_2}|x}\neq 0\}| = m_2-1.
		\end{equation*}
		Thus we get that actually for the given state $\ket{x}$ in the Hilbert space $A_1\otimes A_2\otimes A_3$ the value $m=max\{m_1-1,m_2-1,m_3\}=\tilde{m}-1$. Hence by induction hypothesis $\ket{x}$ has Schmidt decomposition. Thus we have proven the Theorem[\ref{3p}] for Case-II.
	\end{enumerate}
	Thus the Theorem[\ref{3p}] is proven to be true for $m=\tilde{m}$ exhaustively by accounting all possible cases. Hence the Theorem[\ref{3p}] is true $\forall m
	\in \mathbb{N}$ by induction. This completes the proof.
	
\end{proof}


\begin{thebibliography}{99}
	\bibitem{1}  A. K. Pati, Phys. Lett. A \textbf{278}, 118 (2000).
	\bibitem{2} M. A. Nielsen and I. L. Chuang, \textit{Quantum Computation and Quantum Information} (Cambridge University Press, 2010), p. 109, 110.
	\bibitem{3} E. Schmidt, Math. Ann. \textbf{63}, 433 (1907).
	\bibitem{4}  See for example, A. Peres, \textit{Quantum theory: concepts and methods} (Kluwer, Dordrecht, 1993), p. 123.
	\bibitem{5} J. Preskill, \textit{Lecture Notes on quantum Information}, Web page of Preskill.
	\bibitem{6} A. Ac\'{\i}n, A. Andrianov, L. Costa, E. Jan\'{e}, J. I. Latorre and R. Tarrach, Phys. Rev. Lett. \textbf{85}, 1560 (2000).
	\bibitem{7} T. A. Brun and O. Cohen, Phys. Lett. A \textbf{281}, 88 (2001).
	\bibitem{8} H. A. Carteret, N. Linden, S. Popescu, and A. Sudbery, Found. Phys. \textbf{29}, 527 (1999).
	\bibitem{9} H. A. Carteret and A. Sudbery, J. Phys. A \textbf{33}, 1 (2000).
	\bibitem{10} V. Coffman, J. Kundu, and W. K. Wootters, Phys. Rev. A \textbf{61}, 2306 (2000).
	\bibitem{11} M. Grassl, M. R\"{o}tteler, and T. Beth, Phys. Rev. A \textbf{58}, 1853 (1998).
	\bibitem{12} J. Kempe, Phys. Rev. A \textbf{60}, 910 (1999).
	\bibitem{13} H. A. Carteret, A. Higuchi and A. Sudbery, Journal of Mathematical Physics \textbf{41}, 7932 (2000).
	\bibitem{14} N. Linden and S. Popescu, Fortschr. Phys. \textbf{46}, 567 (1998).
	\bibitem{15} Y. Makhlin, arXiv:quant-ph/0002045.
	\bibitem{16} A. Sudbery, J. Phys. A \textbf{34}, 643 (2001).
	\bibitem{17} A. V. Thapliyal, Phys. Rev. A \textbf{59}, 3336 (1999).
	\bibitem{18} W. K. Wootters, Phil. Trans. Roy. Soc. A \textbf{356}, 1717 (1998).

\end{thebibliography}
\end{document}